\documentclass[11pt]{article}%

\usepackage{graphicx}
\usepackage{amsthm}
\usepackage{amsmath}
\usepackage{amsfonts}
\usepackage{amssymb}
\usepackage{color}

\def\LT{{\mathbb{LT}}}
\def\bfK{\mathbf K}

\def\bft{\mathbf t}
\def\ubft{\underline\bft}
\def\bfsigma{{\mbox{\boldmath${\sigma}$}}}

\def\ubfsigma{\underline{\mbox{\boldmath${\bfsigma}$}}}

\def\bfPi{{\mbox{\boldmath${\Pi}$}}}

\newcommand{\N}{{\mathbb{N}}}

\newcommand{\R}{\mathbb{R}}

\newtheorem{theo}{Theorem}

\newtheorem{cor}{Corollary}
\newtheorem{rmk}{Remark}

\newtheorem{prop}{Proposition}

\numberwithin{equation}{section}

\begin{document}

\begin{center}

{ \large \bf Critical region for an Ising model coupled to causal dynamical triangulations}

\vspace{30pt}

{\sl J.~ Cerda-Hern\'andez}$\,^{a}$ 

\vspace{24pt}

{\footnotesize


$^a$~Institute of Mathematics, Statistics and Scientic Computation,\\ 
University of Campinas - UNICAMP,\\
Rua S\'ergio Buarque de Holanda 651, CEP 13083-859, Campinas, SP, Brazil.\\
E-mail: javier@ime.usp.br.
}

\vspace{48pt}

\end{center}

\begin{abstract}
This paper extends results obtained by \cite{HeAnYuZo:2013} for the annealed  Ising model 
coupled to two-dimensional causal dynamical triangulations. We employ the Fortuin-Kasteleyn (FK) representation 
in order to determine a region in the quadrant of parameters $\beta,\mu>0$ where the critical curve for 
the annealed model is possibly located. This is done by outlining a region where the model has a unique 
infinite-volume Gibbs measure, and a region where the finite-volume Gibbs measure does not have weak limit (in fact, 
does not exist if the volume is large enough). We also improve the region of subcritical behaviour of the model, 
and provide a better approximation of the free energy.
\\ \\
\textbf{2000 MSC.} 60F05, 60J60, 60J80.\\
\textbf{Keywords:} causal dynamical triangulation (CDT), Ising model, partition function, Gibbs measure, 
transfer matrix, FK representation
\end{abstract}

\newpage

\section{Introduction. A review of related results} 
Models of planar random geometry appear in physics in the context of two-dimensional quantum gravity 
and provide an interplay between mathematical physics and probability theory.  

Causal dynamical triangulation (CDT) and its predecessor 
dynamical triangulation (DT), introduced by Ambj{\o}rn and Loll (see \cite{Ambjorn:1998xu}),  constitute attemps to provide a  meaning to the formal expressions appearing in the path integral 
quantization of gravity (see \cite{Ambjorn:1997di}, 
\cite{Ambjorn:2006} for an overview). The idea is to approximate the path integral by 
changing integration with respect to continuous random geometry by integration with respect to DTs or CDTs.  
As a result, the path integral with respect to continuous random geometries is replaced by a weighted sum over discrete geometries, 
where the weights are defined  by the discrete form of the original action. In this paper we adopt the framework of  
CDTs coupled with an Ising spin system.

Putting a spin system  on  the collection of all causal triangulations is  interpreted as a {\it coupling  gravity with matter}.
This type of coupling has been a subject of persistent interest in physics once the matrix model methods were 
successfully applied to the Ising model on random lattice with coordination number $4$ (see \cite{Kazakov:1986hu} and 
\cite{Boulatov:1986sb} for details). In this article we discuss the Ising model coupled to causal triangulations in details, 
study the termodinamic function known as free energy, and formally define  the critical curve for the annealed model. 
This annealed model was first  introduced by \cite{Ambjorn:1999gi}. 
This model was also  studied by \cite{Benedetti:2006rv}, where the authors 
develop a systematic counting of embedded graphs and evaluate thermodynamic functions at high and low temperature expansion (see
\cite{Ambjorn:2008jg} and \cite{HeAnYuZo:2013} for other progress made for the annealed model).

The main advantage of the CDTs over the DTs is that the ensemble of CDTs is more regular than the ensemble of the DTs, 
in the sense that does not permite spacial topology changes and leads to a fractal dimension $d_H=2$
(see \cite{Ambjorn:1998xu}, \cite{Ambjorn:1997di}, \cite{Durhuus:2009sm}, \cite{Durhuus:2006}). On the other hand, the causality 
constraints still make it difficult to find an analytical 
solution of the Ising model coupled to CDTs, however
some progress has been recently 
made about existence of Gibbs measures and phase transitions (see \cite{HeAnYuZo:2013}, \cite{anatoli} for details). Employing  transfer 
matrix techniques  and in particular the Krein-Rutman theory of
positivity-preserving operators, 
\cite{HeAnYuZo:2013}  determines  a region in the quadrant of parameters $\beta,\mu >0$ 
where the finite-volume free energy  has a thermodynamic limit. In this article the authors  provide first  approximations  of the 
infinite-volume free energy, and conditions for  existence of Gibbs measures. Furthermore, the authors  prove that in the determined 
region the 
annealed model presents a subcritical behaviour, i.e., there exists a unique Gibbs distribution for the model, and that with respect to 
this limit measure the average surface has a form of an infinite narrow tube,  
which is essentially one dimensional. 

Computation of the critical curve (point) for models in statistical mechanics has  always been a challenging problem, and since FK models were 
introduced  by Fortuin and Kasteleyn in  1969 (see \cite{FK:1972}), they have become an 
important tool in the study of phase transition for spin models. The goal of this article is  to use  the FK representation   of 
the Ising model  coupled to CDT in order to obtain lower and upper bounds for the critical
curve. We employ  the results of \cite{HeAnYuZo:2013}  to establish a region 
where the N-strip Gibbs measure has no weak limit as $N\to\infty$, where the boundary of this region provides a lower bound for the critical 
curve. Futhermore, we expand the region of subcritical behaviour of the model, 
and provide a better approximation to the free energy.
The aforementioned FK representation utilizes a family of  Poisson point processes and the Lie-Trotter product formula to interpret exponential
sums of operators as random operator products.  This representation was originally derived in  \cite{AizKleinNewman} (see \cite{Aizenman}, 
\cite{Ioffe} for an overview). 

Results presented in this paper, along with the results in \cite{HeAnYuZo:2013}, 
can be utilized to determine a region in the quadrant of parameters $\beta,\mu >0$ where 
the critical curve 
is possibly located. Moreover, the FK representation  allows to compute 
lower and upper bounds  for the infinite-volume  free energy.

The paper  is organized as follows. In  Sections \ref{Sect2.1}-\ref{Sect2.3}  we  review the main features of Lorentzian 
CDTs  and describe the Ising model coupled to CDT. In addition, we  define the free energy and provide some important properties of the model. 
In Section \ref{Sect2.4} we present the main results 
(Theorems \ref{theomain1} and \ref{theomain2}). In Section \ref{FKrepresent} we review  the FK representation  of 
Ising model  coupled to CDT. Sections \ref{Sect3.2} and \ref{Sect3.3} contain the proofs of the theorems formulated in 
Section \ref{Sect2.4}.

\section{The notation and results}\label{Sect2}

We start with an introduction to the basic features of CDTs. More details and properties about causal triangulations can be be found in 
\cite{anatoli}, \cite{SYZ1}, \cite{SYZ2}, \cite{Ambjorn:1997di}, \cite{Durhuus:2009sm},\cite{HeAnYuZo:2013}, \cite{MYZ2001}.

\subsection{Two-dimensional Lorentzian models}\label{Sect2.1}

As in \cite{HeAnYuZo:2013}, we will work 
with rooted causal dynamic triangulations of the cylinder 
$C_N = {\mathcal S}\times [0,N]$, $N = 1, 2, \dots$, which 
have $N$ bonds (strips) ${\mathcal S}\times [j,j+1]$.  Here ${\mathcal S}$ stands for
a unit circle (see Figure \ref{fig1}). Formally, a triangulation $\ubft$ of $C_N$ is called a {\it rooted causal dynamic 
triangulation} (CTD) if the following conditions hold:
\begin{itemize}
\item each triangular face of $\ubft$ belongs to some strip $\mathcal S \times [j, j + 1]$, $j =
1, \dots, N-1$, and has all vertices and exactly one edge on the boundary
$(\mathcal S \times \{j\}) \cup (\mathcal S\times \{j+1\})$ of the strip $\mathcal S\times [j, j + 1]$;
\item the number of edges on $\mathcal S \times \{j\}$ should be finite for any $j = 0, 1, \dots, N-1$: let  $n^j = n^j(\ubft)$ be  the number of edges on $\mathcal S \times \{j\}$, then 
$1 \leq n^j < \infty$ for all $j = 0, 1, \dots, N-1$.
\end{itemize}
and have a root face, with the anti-clockwise ordering of its vertices $(x,y,z)$, where $x$ and $y$ lie
in ${\mathcal S} \times\{0\}$.

The CDTs arise naturally when physicists attempt to define a
fundamental path integral in quantum gravity. See \cite{Ambjorn:2006} for a review of the relevant literatute; for 
a rigorous mathematical background of the model, cf. \cite{MYZ2001}. Additional properties of CDTs have been studied in
\cite{SYZ1}.

A rooted CDT $\ubft$  of $C_N$ is identified with a compatible sequence
$$\ubft = (\bft(0), \bft(1), \dots, \bft (N-1)),$$ where $\bft (i)$ is a triangulation of the 
strip ${\mathcal S}\times [i,i+1]$. The compatibility means that 
\begin{equation}\label{compat}
n_{up}(\bft(i+1))=n_{do}(\bft(i)),\quad i= 0,\dots, N-2.
\end{equation}

Note that for any edge lying on the slice $\mathcal S \times \{ i\}$ belongs to exactly two 
triangles: one up-triangle from $\bft (i)$ and one down-triangle from $\bft (i-1)$. This provides 
the following relation: the number of triangles in the triangulation $\ubft$, denoted by $n(\ubft)$, is twice the total number of 
edges on the slices. More precisely, remind that  $n^i$ is the number of edges on slice $\mathcal S \times \{i\}$. 
Then, for any $i=0,1,\dots, N-1$,
\begin{equation}\label{et2-yamb}
n(\bft(i)) = n_{up} (\bft (i)) + n_{do}(\bft (i)) = n^i + n^{i+1}.
\end{equation}

In the physical approach to statistical models, the computation of the partition function is the first step towards 
a deep understanding of the model, enabling for instance the computation of the free energy. Follow this approach, the 
computation of the partition function for the case of pure CDTs, was first introduced and computed in \cite{Ambjorn:1998xu} 
(see also \cite{MYZ2001} for a mathematically rigorous account). 

Considering triangulations with  periodical
spatial boundary conditions, i.e. the strip $\bft (N-1)$ is compatible with $\bft (0)$, and let $\LT_N$ denote the set of causal 
triangulations on the cylinder $C_N$ with this boundary condition, thus the partition function for rooted CDTs in the cylinder $C_N$ with periodical
spatial boundary conditions and for the value of the cosmological constant $\mu$ is given by
\begin{equation} \label{yamb-pf1}
Z_N(\mu)=\sum_{\ubft} e^{-\mu n(\ubft) } = \sum_{(\bft(0), \dots, \ubft(N-1))} \exp \Bigl\{-\mu \sum_{i=0}^{N-1} n(\bft(i)) \Bigr\}.
\end{equation}

Moreover, the periodical spatial boundary condition on the CDTs permits to write the partition function  
$Z_N(\mu)$ in a trace-related form 
\begin{equation}\label{yamb-pf1-2}
Z_N(\mu)= \mbox{tr}\; \bigl( U^N \bigr).
\end{equation}
This gives rise to a transfer matrix $U=\{u(n,n^\prime)\}_{n,n^\prime = 1, 2, \dots}$
describing the transition from one spatial strip to the next one. It is an infinite matrix with
positive entries
\begin{equation}\label{yamb-tmpg}
    u(n,n^\prime) = \binom{n+n^\prime-1}{n-1}e^{-\mu(n+n^\prime)}.
\end{equation}

\begin{figure}[t]
\begin{center}
\includegraphics[height=6cm,width=12cm]{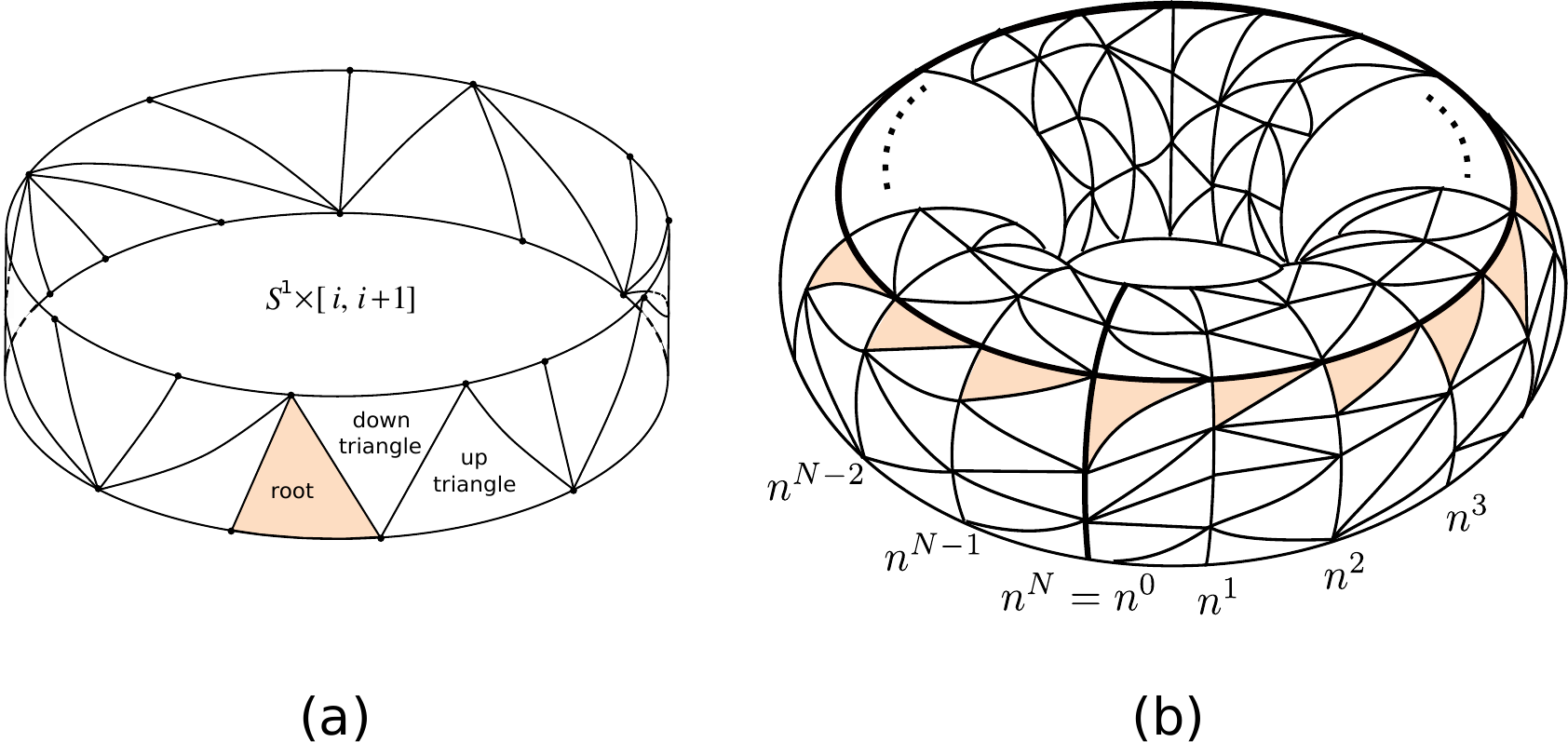}
\end{center}
\caption{(a) A strip triangulation of ${\mathcal S}\times [j,j+1]$. (b) Geometric representation of a CDT
with periodic spatial boundary condition.}%
\label{fig1}
\end{figure}
Employing the $N$-strip partition function for pure CDTs with periodical boundary condition, defined by the formula 
(\ref{yamb-pf1}), we define the  $N$-strip Gibbs probability distribution for pure CDTs
\begin{equation}\label{Q}
 \mathbb{Q}_{N,\mu}(\ubft)=\frac{1}{Z_N(\mu)} \mbox{e}^{-\mu n(\ubft)}.
\end{equation}
The asymptotic properties of the partition function of the model and existence of the weak limit of the measure 
$\mathbb{Q}_{N,\mu}$ for $\mu\geq \ln2$ are provided in \cite{MYZ2001}. These results were obtained by 
utilizing transfer-matrix formalism and tree parametrization of Lorentzian triangulations. In addition, 
the transfer-matrix formalism  and physical considerations suggest  that, as $N\to \infty$, the partition
function is controlled by the largest eigenvalue $\Lambda$ of the transfer matrix (\ref{yamb-tmpg}):
\begin{equation}\label{aproxLambda}
Z_N(\mu) = {\rm tr}\;U^N \sim\Lambda^N,
\end{equation}
where
\begin{equation}\label{Lambda(g)}
\Lambda:= \Lambda(\mu)=\left[ \frac{1- \sqrt{1-4\exp(-2\mu)}}{2\exp(-\mu)}\right]^2.
\end{equation}
That heuristic result was proved in \cite{HeAnYuZo:2013}.

The following properties hold and will be utilized  to  prove the main results.

\noindent {\bf Property 1.} {\rm{(Theorem 1 in \cite{MYZ2001}).}} {\it For any $\mu \geq \ln 2$ 
the following relation holds:
\begin{equation}\label{yamb-e13}
\lim_{N\to\infty}\frac{1}{N}\ln\,Z_N(\mu)=\ln\,\Lambda (\mu).
\end{equation}
Furthermore, the $N$-strip Gibbs measure $\mathbb{Q}_{N,\mu}$ converges 
weakly to a limiting measure $\mathbb{Q}_{\mu}$.}

\noindent {\bf Property 2.} {\rm{(Proposition 5, \cite{MYZ2001}).}} {\it For any 
$\mu <\ln 2$, the $N$-strip partition function $Z_N(\mu)$ exists if }
\begin{equation}\label{yamb-e14}
\mu > \ln \left( 2\cos \displaystyle\frac{\pi}{N+1} \right).
\end{equation}
Another proof  of Property 1 can be found in \cite{HeAnYuZo:2013}. In order to prove this property the authors 
utilize the transfer-matrix formalism and Krein-Rutman theorem. Inequality (\ref{yamb-e14}) in Property 2 implies that
if $\mu<\ln2$, then there exists  $N_0\in\N$ such that  $Z_N(\mu)=\infty$ if $N>N_0$.

\subsection{Ising model on random causal triangulations: Definition and preliminary results}\label{Sect2.3}

Let $\Delta(\ubft)$ and $\Delta(\bft(i))$ denote the set of triangles of the triangulation $\ubft$ and 
the set of triangles of the strip $\bft(i)$, respectively. As in \cite{HeAnYuZo:2013}, 
each triangle from the triangulation $\ubft$ is associated with a spin taking values
$\pm 1$. An $N$-strip
configuration of spins is represented by a collection
$$\ubfsigma =(\bfsigma (0), \bfsigma (1), \dots, \bfsigma (N-1))$$ 
where $\bfsigma (i)\equiv\bfsigma (\bft (i)):=\{\sigma(t)\}_{t\in \Delta(\bft(i))}$ is a configuration of spins 
$\sigma (t)$
over triangles
$t$ forming a triangulation $\bft (i)$. 

We consider a usual ferromagnetic Ising-type energy where two spins $\sigma (t)$ and $\sigma (t')$ interact if their 
corresponding supporting triangles $t$, $t'$ share a common edge. In this case the triangles $t$, $t'$ are referred to 
as nearest neighbors and denoted by $\langle t, t'\rangle$.

Thus, the Hamiltonian utilized for the  annealed model is given by
\begin{equation}\label{IsingHami}
{\mathbf h}(\ubfsigma,\ubft )=-\sum_{\langle t,t' \rangle} \sigma (t)\sigma (t').
\end{equation}

The partition function for the $N$-strip Ising model coupled to
CDT, at the
inverse temperature $\beta >0$ and the cosmological constant $\mu$, is given
by
\begin{eqnarray}\label{yamb-pf}
&& \Xi_N(\beta,\mu)=\sum_{(\bft (0),\dots,\bft (N-1))}
\exp\Bigl\{ -\mu \sum_{i=0}^{N-1} n(\bft (i)) \Bigr\} \\ \nonumber &&
\quad\times \sum_{(\bfsigma (0),\dots ,\bfsigma (N-1))}
\prod_{i=0}^{N-1}
\exp\,\Bigl\{ -\beta h(\bfsigma (i))-\beta v(\bfsigma (i),\bfsigma (i+1))\Bigr\},
\end{eqnarray}
where $n(\bft (i))$ stands for the number of triangles in the triangulation $\bft (i)$, $h(\bfsigma(i))$ denotes
the inner energy of the configuration $\bfsigma (i)$, and $v(\bfsigma (i),\bfsigma (i+1))$ is the energy of 
interaction between the neighboring triangles
belonging to the adjacent strips ${\mathcal S}\times [i, i+1]$ and ${\mathcal S}\times [i+1, i+2]$.

As  in the case of pure CDTs, application of the transfer-matrix formalism to Ising coupled to CDTs suggests that the partition 
function would be rewritten as
\begin{equation}\label{trmatrix}
\Xi_N(\beta,\mu)={\rm tr}\;\bfK^N.
\end{equation}
The  formal definition of the operator $\bfK$ and a proof of the identity (\ref{trmatrix}) can be found in
\cite{HeAnYuZo:2013}. 

Following \cite{HeAnYuZo:2013}, we define the operator $\bfK$ on the Hilbert space $\ell^2_{\rm{T-C}}$ 
where  the subscript T-C refers to triangulations 
and spin-configurations of square-summable functions
$$(\bft,\bfsigma)\mapsto \phi (\bft,\bfsigma),\;\hbox{with}\; \sum\limits_{\bft,\bfsigma}|\phi(\bft,\bfsigma)|^2<\infty,$$
where the argument $(\bft,\bfsigma )$ runs over single-strip triangulations and
supported configurations of spins, with the scalar product
$\left\langle\phi',\phi''\right\rangle_{\rm T-C}= \break
\sum_{\bft ,\bfsigma}\phi' (\bft ,\bfsigma )
{\overline{\phi''}}(\bft ,\bfsigma )$ and the induced norm $\|\phi\|_{\rm T-C}$.

As in the case of the pure CDT,  we  introduce the $N$-strip Gibbs probability distribution
associated with (\ref{yamb-pf}) as
\begin{eqnarray}\label{yamb-Gd}
\mathbb P^{\beta,\mu}_N (\ubft,\ubfsigma) &=& \frac{1}{\Xi_N (\beta,\mu )} e^{-\mu n(\ubft) -\beta {\bf h}(\ubfsigma,\ubft)}.
\end{eqnarray}
Let $\mathcal{G}_{\beta,\mu}$ be  the set of {\it Gibbs measures} given by the closed convex hull of the set
of weak limits:
\begin{equation}
 \mathbb P^{\beta,\mu}=\lim_{N\to\infty} \mathbb P^{\beta,\mu}_N.
\end{equation}
We  define the domain of parameters where the weak limit Gibbs distribution exists, as follows
\begin{eqnarray}
 \Gamma &=&\{(\beta,\mu)\in\mathbb{R}^2_{+}:\; \mathcal{G}_{\beta,\mu}\neq \emptyset \}.
\end{eqnarray}
Further, the critical curve $\gamma_{cr}$  for the Ising model coupled to CDT is defined as
\begin{equation}\label{f_cr}
 \gamma_{cr}=\partial\Gamma \cap \mathbb{R}^2_+.
\end{equation}
This critical curve is of major importance, since it splits the parameter space  into two following regions: a 
region where the Gibbs distributions exist, and a region where the Gibbs distribution does not exist, due to 
large Boltzman  factors associated with the triangulations which results in an infinite  partition function.
This behaviour is observed at zero and infinite temperature in 
the annealed model.

In order to study  behaviour of the critical curve we focus on the thermodynamic function 
which is defined as
\begin{equation}
\phi(\beta,\mu)=\lim_{N\to\infty}\phi_N(\beta,\mu)=\displaystyle\frac{1}{N}\ln \Xi_N (\beta,\mu).
\end{equation}
Note that, 
$$
\begin{array}{ccl}
\phi_N(0^+,\mu)=\displaystyle\lim_{\beta\to0^+}\phi_N(\beta,\mu)&=&\displaystyle\frac{1}{N}\ln \left(\displaystyle\sum_{\bft} \exp\left\{-\mu\sum_{i=0}^{N-1}n(\bft(i)) \right\} \displaystyle\sum_{\bfsigma} 1 \right)\\
             &=& \displaystyle\frac{1}{N}\ln Z_N(\mu-\ln2).
\end{array}
$$
Utilizing (\ref{yamb-e13}), it is easy to show that $\phi_N(0^+,\mu)$ converges  to $\ln\Lambda(\mu-\ln2)$ as $N\to \infty$. This result implies 
that the free energy exists at infinitely high temperature  for $\mu>2\ln2$,
$$\phi(0^+,\mu)=\lim_{N\to \infty}\phi_N(0^+,\mu)=\ln\Lambda(\mu-\ln2).$$

It can be noticed that,  if  $\beta=0$, and   $\mu\to 2\ln2$,  then  $\phi(0^+,\mu)\to 0$. Thus, utilizing the results 
(\ref{yamb-e13}) and (\ref{yamb-e14}) (see \cite{MYZ2001} for details), we find the infinite-volume free energy 
$\phi(0^+,\mu)$ at infinitely high temperature 
\begin{equation}\label{freeEin0}
\phi(0^+,\mu)=\left\{
\begin{array}{ccl}
\ln\Lambda(\mu-\ln2)<+\infty &,& \mu(0^+)>2\ln2\\
0    &,& \mu(0^+)=2\ln2\\
+\infty &,& \mu(0^+)< 2\ln2\\
\end{array}\right.
\end{equation}
In addition, if  the temperature tends to infinity, we obtain that the annealed model is  a pure CDT model with a critical
parameter $\mu_{cr}(0^+)=2\ln2$. Furthermore, for all $\mu\geq 2\ln2$ the Gibbs measure is unique, and as in the case of pure CDTs 
there exist two types of triangulations with respect to this measure. 
If $\mu>2\ln2$ we obtain subcritical triangulations. Utilizing the identity  (\ref{freeEin0}) and the results obtained in 
\cite{Durhuus:2006} and  \cite{MYZ2001} it can be showed that in this case the triangulations have a bounded diameter with probability 1,
which is essentially one dimensional random geometry. The boundary value  $\mu=2\ln2$  yields a two dimensional random geometry with 
Hausdorff dimension $d_H=2$ (see \cite{Durhuus:2009sm}, \cite{Durhuus:2006} and \cite{SYZ2}  for details).\\ 

In this paper we show that for any $0<\beta<\infty$ the infinite-volume free energy of the annealed model  has the same 
behaviour  as (\ref{freeEin0})  except for a small interval $I(\beta)=[f_1(\beta),f_2(\beta)]$ where the behaviour of the function 
is unknown. On the other hand, according to the numerical results (see \cite{Ambjorn:1999gi}, 
\cite{Ambjorn:2008jg}, \cite{Benedetti:2006rv}) there exists a unique value $\widetilde{\mu}_{cr}=
\widetilde{\mu}_{cr}(\beta)\in I(\beta)$  such that 
$\phi(\beta,\mu)=+\infty$ for $\mu<\widetilde{\mu}_{cr}(\beta)$,  $\phi(\beta,\mu)=0$ for $\mu=\widetilde{\mu}_{cr}(\beta)$, and 
$\phi(\beta,\mu)<+\infty$ for $\mu>\widetilde{\mu}_{cr}(\beta)$, however the exact value of $\widetilde{\mu}(\beta)$ 
has not been computed yet. The interval $I(\beta)$ is computed explicitly in the next section.\\

Note that, if $0<\beta<\infty$, for any triangulation $\bft$ the energy of any spin configuration $\bfsigma$ on 
$\bft$ is  larger than or equal to the energy of a pure configuration (all $+$'s or all $-$'s). i.e., 
${\bf h}(\bfsigma,\bft)\geq -3/2 n(\bft)$. Thus, for any $\beta>0$ the following inequality holds
$$\phi_N(\beta,\mu)< \displaystyle\frac{1}{N}\ln \left( \sum_{\bft}e^{-(\mu-\frac{3}{2}\beta -\ln2 )n(\bft)} \right)
=\displaystyle\frac{1}{N}\ln Z_N\left(\mu-\frac{3}{2}\beta -\ln2\right).$$
Therefore, 
\begin{equation}\label{BD}
\phi(\beta,\mu)=\lim_{N\to \infty}\phi_N(\beta,\mu) \leq\ln\Lambda\left(\mu-\frac{3}{2}\beta -\ln2\right).
\end{equation}
Hence, the 
inequality $\mu> \frac{3}{2}\beta+2\ln2$ provides a sufficient condition for subcritical behaviour (uniqueness of Gibbs measure)
of the Ising model coupled to CDTs, and provides a region where the infinite-volume free energy exists and is finite, i.e., 
does not have singularities.

Now, based on (\ref{yamb-e13})  and (\ref{BD}) an upper bound for the critical curve of the annealed model is given by  
\begin{equation}\label{AM}
\mu< \frac{3}{2}\beta+2\ln2\;\;\mbox{for all}\;\; (\beta,\mu)\in \gamma_{cr}.
\end{equation}

Obviously, the upper bound given by (\ref{AM}) is not sharp for all $\beta$. 
For example, for the low values of $\beta$, which correspond to high temperature, the upper bounds provided by  
\cite{HeAnYuZo:2013} and \cite{Ambjorn:1999gi} are more accurate. 
In this paper we focus on finding  lower and upper bounds for the critical curve for all values of $\beta$.

\subsection{The main results}\label{Sect2.4}
In the previous subsection we examine basic notions of statistical mechanics of the model and review some properties about 
the phase space, the free energy and Gibbs distributions. In this section we focus on presenting the main results of this paper. 
Utilizing these results we provide a first phase diagram for the Ising Model coupled to CDTs in 
the parameters $\beta$ and $\mu$, and 
compute an approximation of the infinite-volume free energy. The exact phase diagram at infinite temperature was 
computed in the previous subsection.

In order to construct an upper bound for the critical curve, we define the function  $\lambda (\beta,\mu)$  as follows
\begin{equation}\label{lambda}
\lambda (\beta,\mu)= c^2\,(m^2+1)\,({\rm cosh}\,2\beta)
\left( 1+ \sqrt{1-\frac{1}{({\rm cosh}\,2\beta)^2}\frac{(m^2 -1)^2 }{(m^2 + 1)^2}}\right),
\end{equation}
where $c$ and $m$ are determined by
\begin{eqnarray*}\label{ccccc}
c&=&\frac{\exp(\beta-\mu )}{e^{2\beta}(1-\exp(\beta-\mu ))^2 - e^{-2\mu}}, \\
\label{mmmmmm}
m &=& e^{2\beta} + (1-e^{4\beta})\exp\,(-(\beta +\mu) ).
\end{eqnarray*}
Let $\psi(\beta)$ be  a strictly increasing function defined as
\begin{equation}\label{functpsi}
\psi(\beta) = \inf\{ \mu\in\mathbb{R}^+ : \lambda(\beta,\mu)< 1  \},\quad\mbox{for}\quad \beta>0.
\end{equation}
It is easy to see that  
\begin{equation}\label{beta=0}
\psi(0)=\lim_{\beta \to 0^+} \psi(\beta)=2\ln2.
\end{equation}
In addition, it is well known   that in the 
region $\{(\beta,\mu) : \mu>  \psi(\beta) \}$ the free energy $\phi(\beta,\mu)$ is finite. Consequently, in this region the 
annealed model has a unique Gibbs measure (see \cite{HeAnYuZo:2013} for details).

Let   $\bft_{1},\dots,\bft_{k}$ be triangulations of a single strip $\mathcal S\times [0,1]$ and
$\bfsigma_{1},\dots,\bfsigma_{k}$ be their corresponding spin configurations. Given $0\leq i_1 <\cdots < i_k\leq N-1$,
we define 
a finite-dimensional  cylinder $\mathcal{C}_{i_1,\dots,i_k}= \mathcal{C}^{(\bft_1,\bfsigma_1),\dots,(\bft_k,\bfsigma_k)}_{i_1,\dots,i_k}$ as follows
\begin{equation}\label{cylinder}
\mathcal{C}_{i_1,\dots,i_k}=\left\{(\ubft,\ubfsigma)\;:\;(\bft(i_1),\bfsigma(i_1))= 
(\bft_{1},\bfsigma_{1}),\dots,(\bft(i_k),\bfsigma(i_k))=(\bft_{k},\bfsigma_{k})\right\}.
\end{equation}
\begin{theo}\label{theomain1}
If $(\beta,\mu)\in\mathbb{R}^2_{+}$, where
$$\mu < \max\left\{2\ln 2, 
\displaystyle\frac{3}{2}\ln(2\sinh\beta)+\ln2 \right\},$$
then there exist $N_0\in\N$ such that if $N > N_0$ the partition function $\Xi_N(\beta,\mu)=\infty$ . 
Moreover, the Gibbs distribution  $\mathbb{P}^{\beta,\mu}_N$ with periodic boundary conditions cannot be defined 
by utilizing the standard formula with a normalising denominator $\Xi_N(\beta,\mu)$, consequently, there is 
no limiting probability measure $\mathbb{P}^{\beta,\mu}$ as $N\to\infty$. 
\end{theo}
\noindent Formally Theorem \ref{theomain1} states that, 
for any  finite-dimensional  cylinder  $\mathcal{C}_{i_1,\dots,i_k}$ we obtain 
$\mathbb{P}^{\beta,\mu}_N (\mathcal{C}_{i_1,\dots,i_k}) = 0$  for  $N>N_0\geq\max\{i_1,\dots,i_k\}$.\\

Let $\beta^*_1, \beta^*_2$ be positive solutions of the following equations 
\begin{equation}\label{beta1}
 2\ln 2 = \displaystyle\frac{3}{2}\ln(2\sinh\beta)+\ln2,
\end{equation}
\begin{equation}\label{beta2}
\quad \frac{3}{2}\beta + 2\ln 2 = \psi(\beta),
\end{equation}
respectively. Theorem \ref{theomain1}, along with the results in \cite{HeAnYuZo:2013}, provides two-side bounds for the 
critical curve.
 
\begin{theo}\label{theomain2} The critical curve $\gamma_{cr}$ satisfies the following inequalities.
\begin{enumerate}
 \item[(i)] If $(\beta,\mu)\in \gamma_{cr}$ and $0< \beta < \beta^*_1$, then  
 $$ 2\ln 2 \leq  \mu < \psi(\beta).$$
 \item[(ii)] If $(\beta,\mu)\in \gamma_{cr}$ and  $\beta^*_1 \leq \beta <\beta^*_2$, then  
 $$ \displaystyle\frac{3}{2}\ln(2\sinh\beta)+\ln2 \leq \mu < \psi(\beta).$$
 \item[(iii)] If $(\beta,\mu)\in \gamma_{cr}$ and   $\beta^*_2 \leq \beta <\infty$, then  
 $$ \displaystyle\frac{3}{2}\ln(2\sinh\beta)+\ln2 \leq \mu < \displaystyle\frac{3}{2}\beta +2\ln 2.$$
\end{enumerate}
\end{theo}
\begin{rmk}
Utilizing  (\ref{beta=0}) and $(i)$ of Theorem \ref{theomain2}, we obtain the exact value of the 
critical curve at infinite temperature, i.e., $(0,2\ln2)\in \gamma_{cr}$. Employing the results of \cite{MYZ2001} we conclude
that the set of Gibbs measures at infinite temperature is a single point for $\mu\geq 2\ln2$, and is empty for $\mu< 2\ln2$.  

\end{rmk}
As a by-product of the proof of Theorems \ref{theomain1} and \ref{theomain2}, we obtain a lower and upper bound 
for the infinite-volume free energy. These bounds are provided in the following corollary.\\

\begin{figure}[t]
\begin{center}
\includegraphics[width=10cm]{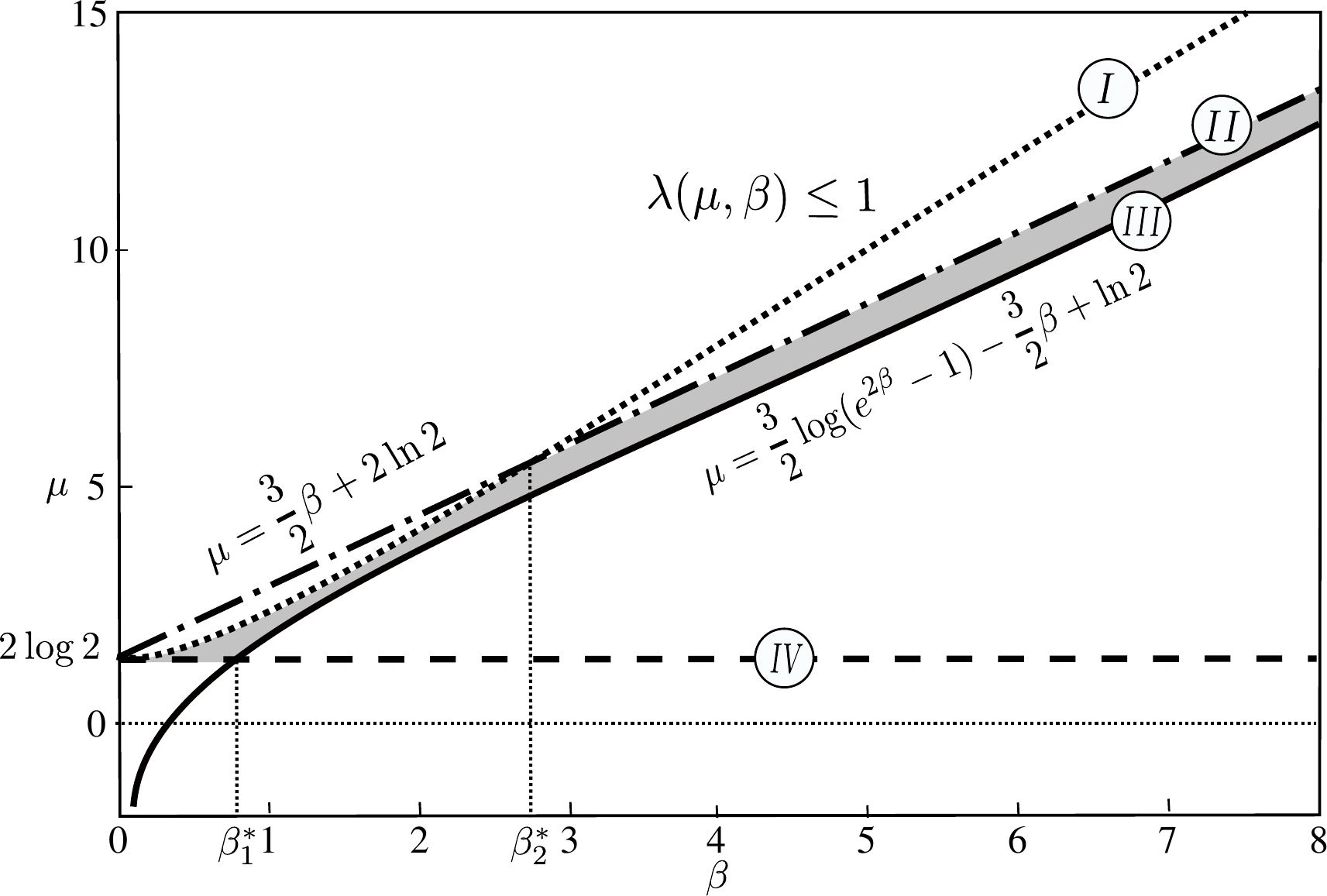}
\end{center}
\caption{The area above the minimum of the dotted curve I (graph of the function
$\psi$ defined in (\ref{functpsi})) and dash-dotted line II is where the limiting
Gibbs probability measure exists and is unique. The critical curve lies in the region below the dotted 
curve I and dash-dotted  line II but above the continuous curve III and dashed line IV.}
\label{fig2} 
\end{figure}

Let $I(\beta)$ denote the interval  $[f_1(\beta),f_2(\beta)]$, where 
\begin{equation}\label{f_1}
f_1(\beta)= \max\left\{ \ln2,  \displaystyle\frac{3}{2}\ln(2\sinh\beta)\right\},
\end{equation}
\begin{equation}\label{f_2}
f_2(\beta)= \min\left\{\psi(\beta)- \ln2,  \displaystyle\frac{3}{2}\beta+\ln2\right\},
\end{equation}
\begin{cor}\label{mainresult3} 
For $(\beta,\mu)$ such that $\mu >  f_2(\beta)$,  the free energy $\phi(\beta,\mu)$ for the 
Ising model coupled to CDTs is finite and satisfies the following inequalities,
  
 $$\ln\Lambda\left(\mu -f_1(\beta) \right)\leq \phi(\beta,\mu) 
 \leq\ln\Lambda\left(\mu - f_2(\beta)\right).$$
Here $\Lambda(s)$ is given by  (\ref{Lambda(g)}).
\end{cor}

\noindent Finally, we obtain the behaviour of the free energy  for the annealed model.

\begin{cor}\label{mainresult4} 
For $\beta>0$ fixed, the free energy $-\phi(\beta,\mu)$ is an increasing function in $\mu$, for all $\mu>f_2(\beta) +\ln2 $ and  
$\phi(\beta,\mu)=\infty$ for $\mu<f_1(\beta) +\ln2 $. In addition, $\lim_{\beta\to0^+}I(\beta)= \{\ln2\}$ and 
$\phi(0^+,\mu)=\ln \Lambda(\mu-\ln2)$.
\end{cor}

\noindent Remember that the exact solution of the free energy at $\beta=0$ was computed in (\ref{freeEin0}). 
The result presented in Corollary \ref{mainresult4} is illustrated in  Figure \ref{freeEnergy}.

\begin{figure}[h!]
\begin{center}
\includegraphics[width=10cm]{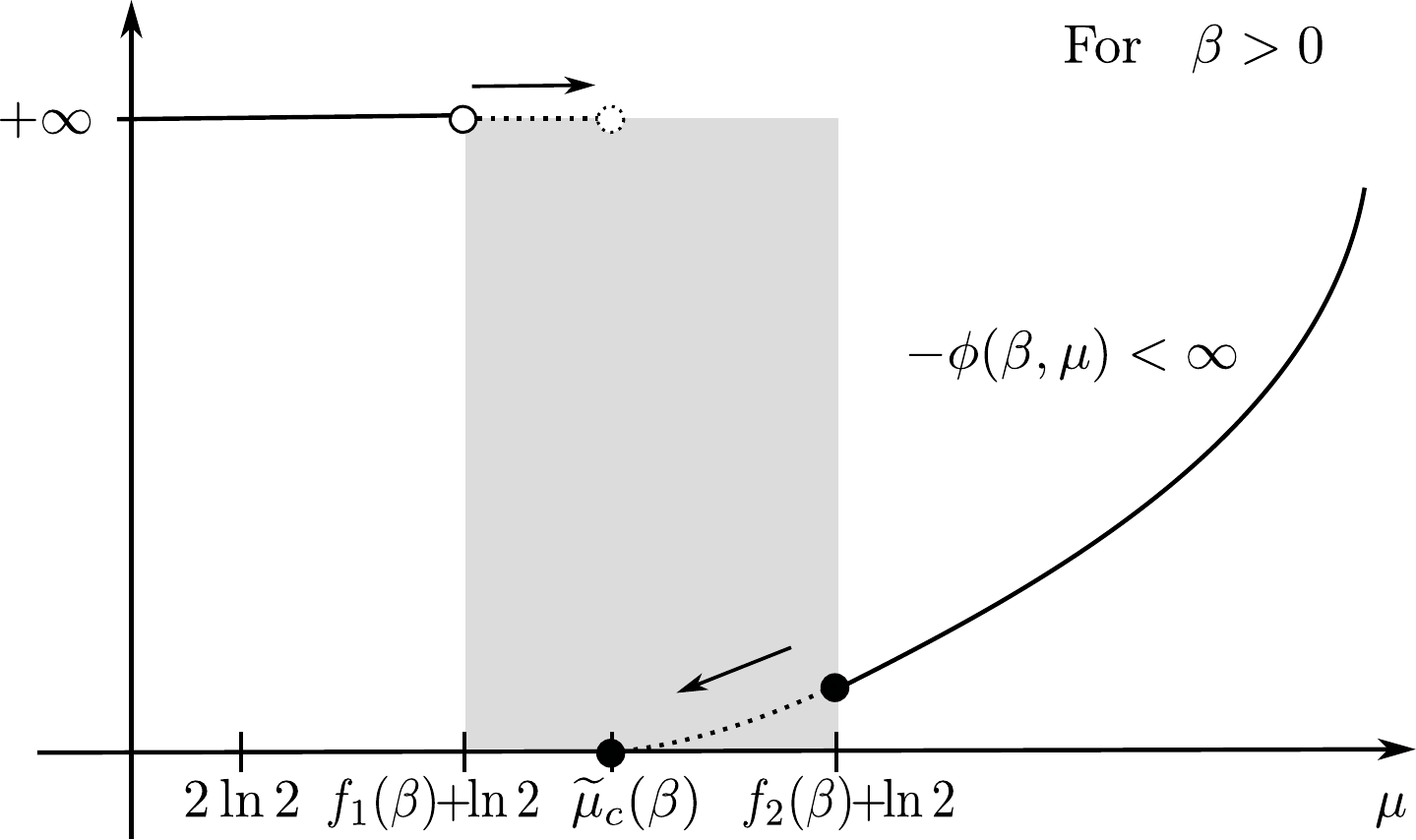}
\end{center}
\caption{Behaviour of the free energy $\phi(\beta,\mu)$ for $\beta$ fixed. In the gray region we show a possible behaviour of 
analytic continuation of the free energy on the right and left side,  according to the numerical results given in \cite{Ambjorn:1999gi}, 
\cite{Ambjorn:2008jg}, \cite{Benedetti:2006rv}.}
\label{freeEnergy} 
\end{figure}

\section{Proof of Theorem \ref{theomain1} and \ref{theomain2} }\label{Sect4}

We start this section with a brief review of the existing results for representation of the  Ising model in terms of 
a two continuum random-cluster model. We apply the results obtained in \cite{AizKleinNewman} and \cite{Ioffe} 
to  Lorentzian triangulations with periodical boundary condition.

\subsection{FK representation for Ising model coupled to CDT}\label{FKrepresent}
As mentioned above, each triangle from the triangulation $\ubft$ is associated with a spin taking values $\pm 1$. This is equivalent
to associating spins on each vertex of the dual Lorentzian triangulation. Starting from this section we shell use the same 
notation $\ubft$ for the dual Lorentzian triangulation 
with a periodical boundary condition.
Note that the Ising model defined in Subsection \ref{Sect2.3}  was defined on dual triangulation.

Let $\mathcal{Z}_N^{\beta,\ubft}$ be the partition function of the Ising model on a fixed dual Lorentzian triangulation $\ubft$, at 
inverse temperature $\beta>0$
\begin{equation}\label{2.20}
\mathcal{Z}_N^{\beta,\ubft}=\sum_{\ubfsigma\in\Omega_{\ubft}} \exp\{ -\beta {\mathbf h}(\ubfsigma,\ubft) \}.
\end{equation}
Here ${\mathbf h}(\ubfsigma,\ubft)$ represents the energy of the configuration $\ubfsigma\in\Omega_{\ubft}$, defined by
(\ref{IsingHami}). Thus, utilizing (\ref{2.20}) the partition function for the $N$-strip Ising model coupled to
CDT, at the inverse temperature $\beta>0$  for the cosmological constant $\mu$,  can be 
rewritten as 
\begin{equation}\label{2.29}
 \Xi_N(\beta,\mu)=\sum_{\ubft}\mbox{e}^{-\mu n(\ubft)} \mathcal{Z}_N^{\beta,\ubft}.
\end{equation}

For any dual Lorentzian triangulation   with periodical boundary condition $\ubft$ and  any  edge 
$e=\langle i,j\rangle$ of $\ubft$,   we select a Poisson process $\xi_e$ of points in $\{e\}\times\R$ with intensity $2$, such
that the members of the family of Poisson processes $\{\xi_e : e\in E(\ubft)\}$  are independent of each other,
where $E(\ubft)$ denotes the set of 
edges of the dual Lorentzian triangulation $\ubft$. Note that $|E(\ubft)|=\frac{3}{2}n(\ubft)$ for all $\ubft\in\LT_N$.  

Let $\mathbb{P}_{\beta,\ubft}$ denote the probability associated with the family of
Poisson processes on the interval $[0,\beta]$. An arrival of $\xi_{e=\langle i,j\rangle}$ at time $t$ 
can be geometrically interpretated as a link between  neighbor vertices 
$i$, $j$ at time $t$. We say that two vertices $t, t'$ of  $\ubft$ (not necessarily neighbor vertices)  are 
connected, if and only  if there exists a path, in the sense of continuum
percolation, connecting $t$ and $t'$ (see  \cite{Grimmett:2008} for  an overview). This property will be 
denoted by $t \leftrightarrow t'$.
The relation $\leftrightarrow$ generates a partition of the set
of vertices into clusters. Given a realization $\xi$
of the Poisson processes $\{\xi_e : e\in E(\ubft) \}$ on $[0,\beta]$, we denote by $k(\xi)$ the number of clusters in 
the realization $\xi$.
 
Papers \cite{Aizenman}, \cite{AizKleinNewman} and \cite{Ioffe} derived a representation  of the 
Ising model in terms of a continuum random-cluster model. Aplication of this representation 
to an Ising model whithout external magnetic field on  a fixed dual Lorentzian triangulation $\ubft$, 
yields the following representation for the partition function $\mathcal{Z}_N^{\beta,\ubft}$.
\begin{prop}\label{prop_Aizenman}
Let $\ubft\in\LT_N$ and $\beta>0$. Then the following identity holds
\begin{equation}\label{eq3.18}
\mathcal{Z}^{\beta,\ubft}_N = \mbox{e}^{\frac{3}{2}\beta n(\ubft)}\int 2^{k(\xi)} \mathbb{P}_{\beta,\ubft}(d\xi).
\end{equation}  
\end{prop}
Utilizing the $N$-strip Gibbs probability distribution $\mathbb{Q}_{N,\mu}$, defined in (\ref{Q}) for the pure CDT case, 
(\ref{2.29}) and Proposition \ref{prop_Aizenman}
we obtain the FK representation of the partition function for the annealed model, at 
the inverse temperature $\beta>0$ and 
for the cosmological constant $\mu$
\begin{equation}\label{FK_ISING}
\Xi_N(\beta,\mu) = Z_N(r)\sum_{\ubft\in\LT_N} \left\{ \int 2^{k(\xi)} \mathbb{P}_{\beta,\ubft}(d\xi) \right\} \mathbb{Q}_{N,r}(\ubft),
\end{equation}
where $r=\mu -\frac{3}{2}\beta$ and $Z_N(\cdot)$ is defined by (\ref{yamb-pf1}).
The result (\ref{FK_ISING}) provides asymptotic behaviour for the free energy, for large $\beta$, 
as follows. 

Let $\ubft=(\bft(0), \bft(1), \dots, \bft (N-1))$ be a Lorentzian triangulation. 
If $\beta$ is large enough then the probability that all the Poisson processes $\xi_e$, $e\in E(\ubft)$ had at least 
one arrival in the time interval $[0,\beta]$ tends to 1, which implies that almost all  vertices are linked, that is, 
$\mathbb{P}_{\beta,\ubft}(k(\xi)=1) \approx 1$.  
Thus, utilizing the representation (\ref{FK_ISING}) we obtain  that 
$$\Xi_N(\beta,\mu) \approx 2 Z_N(\mu-\frac{3}{2}\beta)$$
Consequently, the critical inequality 
$$\mu > \frac{3}{2}\beta + \ln 2$$
gives a necessary (and probability tight) criticality condition for the annealed model. 
Therefore, for  any $\beta$, which is sufficiently large, the free 
the free energy have the following behaviour,
$$\phi(\beta,\mu) \approx \ln \Lambda\left(\mu - \frac{3}{2}\beta\right).$$
Furthermore, this asymptotic property establish that for large 
$\beta$
$$d\left(\gamma_{cr}, \left\{\left(\beta,\frac{3}{2}\beta+\ln2\right) : \beta\;\mbox{large enough} \right\}\right)\approx 0,$$
where $d(\cdot,\cdot)$ is the Euclidean distance  between two sets. A similar heuristic analysis  was 
presented and confirmed by numerical simulations in \cite{Ambjorn:1999gi}. 

\subsection{Trivial lower bound for the critical curve}
The purpose of this section is to study the behaviour of the partition function $\Xi_N$  as a function of $\beta$, 
for  $\beta\geq 0$, $\mu>2\ln2$, $N\in\N$. We further utilize the properties of this function in order to obtain
a lower bound (not sharp) for the critical curve.

\begin{prop}\label{prop_positive}
If  $\beta_1 > \beta_2>0$, then 
$$\mathcal{Z}_{N}^{\beta_1,\ubft} \geq \mathcal{Z}_{N}^{\beta_2,\ubft},$$
for any $\ubft\in \LT_N$.
\end{prop}
\begin{proof}
Utilizing the definition of the partition function  $\mathcal{Z}_{N}^{\beta,\ubft}$, it is easy to obtain the following
expression
$$
\begin{array}{ccl}
\displaystyle\frac{\partial \mathcal{Z}_{N}^{\beta,\ubft}}{\partial \beta} &=& \displaystyle\sum_{\ubfsigma} 
                      \displaystyle\sum_{\langle t,t^\prime \rangle} 
                      \sigma(t)\sigma(t^\prime) e^{-\beta \mathbf{h}(\ubfsigma,\ubft)}
                      = \mathcal{Z}_{N}^{\beta,\ubft} \displaystyle\sum_{\ubfsigma} \displaystyle\sum_{\langle t,t^\prime \rangle} 
                      \sigma(t)\sigma(t^\prime) \mu_{\beta}^{\ubft}(\ubfsigma),
\end{array}
$$
where $\mu_{\beta}^{\ubft}$ denotes the Gibbs measure of an Ising model on $\ubft$, given by
$$
\mu^{\bft}_{\beta} (\ubfsigma)= \frac{1}{\mathcal{Z}_{N}^{\beta,\ubft}}
\exp\Bigl\{-\beta\mathbf{h}(\ubfsigma,\ubft) \Bigr\},
$$
where ${\mathbf h}(\ubfsigma,\ubft)$ is defined by (\ref{IsingHami}). Thus, implementation of the first Griffiths inequality 
(see \cite{Griffiths1972} and \cite{Liggett1985} for an overview) yields
$$\displaystyle\frac{\partial \mathcal{Z}_{N}^{\beta,\ubft}}{\partial \beta}=\mathcal{Z}_{N}^{\beta,\ubft} 
\displaystyle\displaystyle\sum_{\langle t,t^\prime \rangle}\int \sigma_{\{t,t'\}} \mu_{\beta}^{\ubft}(\ubfsigma) \geq 0.
$$
This completes  the proof.
\end{proof}

Now, in particular,  $\mathcal{Z}_{N}^{\beta,\ubft}\geq \mathcal{Z}_{N}^{0,\ubft}=e^{n(\bft) \ln2}$ for all $\beta>0$. 
Thus, the following inequality can be derived 
\begin{equation}\label{triv_lowerbound}
\Xi_N(\beta,\mu) \geq Z_N(\mu-\ln2).
\end{equation}
Utilizing property (\ref{yamb-e14}) and (\ref{triv_lowerbound}), we obtain the trivial lower bound for the critical curve,
\begin{equation}
\mbox{if}\;\;(\beta,\mu)\in \gamma_{cr}, \;\;\mbox{then}\;\;  \mu>2\ln2,
\end{equation}
(see  Figure \ref{fig2}). In addition, it follows from (\ref{triv_lowerbound}) that the free energy satisfies the inequality $\phi(\beta,\mu)\geq \ln\Lambda(\mu -\ln2)$, for $\beta\geq 0$ and $\mu\geq 2\ln2$.  
Finally,  this lower bound along  with the  results in \cite{HeAnYuZo:2013} prove that 
$(0,2\ln2)\in \gamma_{cr}$, and that $\phi(0^+,\mu)= \ln\Lambda(\mu -\ln2)$ for all $\mu\geq 2\ln2$.

\subsection{Proof of Theorem \ref{theomain1}}\label{Sect3.2}
Let $\ubft$ be a dual Lorentzian triangulation on the cylinder $C_N$. 
Given $1\leq  k \leq n(\ubft)$, we define  the set of realizations $\xi$ which splits the set of vertices in $k$ clusters,
\begin{equation}\label{lambda_k}
\bfPi_k=\{\mbox{all realization}\; \xi \; 
\mbox{of}\; \{\xi_{\langle t,t'\rangle}\}\,\mbox{such that}\; k(\xi)= k  \}.
\end{equation}
Thus, we obtain the following representation of (\ref{eq3.18}) 
\begin{equation}\label{eq3.19}
\mathcal{Z}^{\beta,\ubft}_N = \mbox{e}^{\frac{3}{2}\beta n(\ubft)}\sum^{n(\ubft)}_{k=1} 2^{k} \mathbb{P}_{\beta,\ubft}(\bfPi_k).
\end{equation}

Let $\xi\in\bfPi_k$ and let  $\{\mathcal{C}_l\}^{k}_{l=1}$ be   the corresponding cluster decomposition of the  
set $\Delta(\ubft)$. In fact, each cluster $\mathcal{C}_l$ is a subgraph of $\ubft$ formed by vertices $V_l$ and edges $E_l$. 
Note that  $\mathcal{C}_l$ is a random variable and that the cluster decomposition  
$\{\mathcal{C}_l\}^{k}_{l=1}$ can  include isolated vertices. 
Denote by $\eta_l=|V_l|$ and $\kappa_l=|E_l|$,  the number of vertices (triangles) in 
cluster $\mathcal{C}_l$ and the number of edges in $\mathcal{C}_l$, respectively. 
Note that for any decomposition $\{\mathcal{C}_l\}^{k}_{l=1}$,  $\kappa_l$ and $\eta_l$ depend  on 
the geometry of the cluster $\mathcal{C}_l$. Note also that $\sum_{l=1}^k \eta_l = n(\ubft)$.

Now, denote by $\pi(\ubft)$ the set of all maximal unordered partitions of  $\Delta(\ubft)$, i.e. an element
of $\pi(\ubft)$ is an unordered $n$-ple $\{\mathcal{C}_1=(V_1,E_1),\dots,\mathcal{C}_n=(V_n,E_n)\}$ of maximal connected subgraphs 
$\mathcal{C}_i$, with 
$1\leq n\leq n(\ubft)$, such that for any $i,j\in I_n=\{1,2,\dots,n\}$, $V_i\subset \Delta(\ubft)$,
$V_i\neq \emptyset$, $V_i \cap V_j=\emptyset$ and $\cup V_i=\Delta(\ubft)$. A graph $\mathcal{C}_i$ is maximal in the 
sense that: if $t,t'$ are nearest neighbor vertices in $\mathcal{C}_i$, then $\{t,t'\}\in E_i$. 
Given a subgraph $\mathcal{C}$ belonging to some element of $\pi(\ubft)$, we define the set of 
spanning subgraphs  of $\mathcal{C}$ as 
$$Span(\mathcal{C})=\{\gamma: \gamma\; \mbox{is a connected subgraph of}\; \mathcal{C}\;\;\mbox{with}\;\;|V(\gamma)| = |V(\mathcal{C})|\},$$
that is, if $\gamma\in\mathcal{C}$, then $\gamma$ and $\mathcal{C}$ have exactly the same vertex set. 

Finally, the probability that two  nearest neighbor vertices  $t,t'$ are linked is  given by 
${\mathbb P}_{\beta,\ubft} ( t\leftrightarrow t')=1-\mbox{e}^{-2\beta}$.  
Then, denoting by $p={\mathbb P}_{\beta,\ubft} ( t\leftrightarrow t')$, we can rewrite the partition function
$\mathcal{Z}^{\beta,\ubft}_N$ as follows

\begin{equation}\label{newformpf}
\mathcal{Z}^{\beta,\ubft}_N = \mbox{e}^{\frac{3}{2}\beta n(\ubft)}(1-p)^{\frac{3}{2}n(\ubft)} \sum^{n(\ubft)}_{k=1} 2^{k} 
\sum_{\{\mathcal{C}_1,\dots,\mathcal{C}_k\}\in \pi(\ubft)}\rho(\mathcal{C}_1)\cdots \rho(\mathcal{C}_k),
\end{equation}
where 
$$\rho(\mathcal{C}) = 
 \sum_{\gamma\in\; Span(\mathcal{C})} \left(\displaystyle\frac{p}{1-p} \right)^{|E(\gamma)|}.$$
Here $|E(\gamma)|$ denotes the number of edges in the subgraph $\gamma$.
Notice that $|E(\gamma)|\geq 0$ for all $\gamma\in\mathcal{C}$.

A convenient expansion parameter for our analysis is  $u=\frac{p}{1-p}\in [0,\infty)$. 
We now are interested in deriving expansions of the function $\rho(\mathcal{C})$, 
$$\rho(\mathcal{C}) =  \sum_{\gamma\in\; Span(\mathcal{C})} u^{|E(\gamma)|},$$
for values $u>1$ and $u<1$.  
Note that by the definition of a spanning subgraph, it is easy to see that for any $\gamma\in Span(\mathcal{C})$ the following inequalities 
hold 
\begin{equation}\label{inequalityclus}
|\mathcal{C}|-1 \leq |E(\gamma)| \leq \displaystyle\frac{3}{2}|\mathcal{C}|-1,\;\; 
\mbox{if}\;\; \mathcal{C}\subsetneq \Delta(\ubft).
\end{equation}
If $\mathcal{C}=\ubft$  the unique spanning subgraph $\gamma$ that do not satisfy 
the right-hand side of inequality (\ref{inequalityclus}) is
$\gamma=\mathcal{C}$. In this case, the graph $\gamma$ has $|V(\gamma)|=n(\ubft)$ vertices and  
$|E(\gamma)|=\frac{3}{2}n(\ubft)$ edges. If $\gamma\subsetneq \mathcal{C}=\ubft$ the inequality (\ref{inequalityclus}) is satisfied.

In order to obtain lower bounds for the representation (\ref{newformpf}), we implement inequality (\ref{inequalityclus}) for two
different cases: $u > 1$ and  $u < 1$.

\vspace{0.5cm}

\noindent{\it The case $u > 1$ $\big($ equivalent to $\beta>\frac{\ln2}{2}$ $\big)$}. Utilizing
the inequality (\ref{inequalityclus}), we obtain that the following inequality holds for $\gamma\neq \ubft$
$$u^{|\mathcal{C}|-1}\leq u^{|E(\gamma)|}\leq  u^{\frac{3}{2}|\mathcal{C}| -1}.$$
Thus, if $\{\mathcal{C}_1,\dots,\mathcal{C}_k\}\in\pi(\ubft)$ and $k\geq 2$ (the case $k = 1$ will treate separately later), we obtain 
\begin{equation}\label{firstineq}
u^{n(\ubft)-k}\prod_{i=1}^kf(\mathcal{C}_i)\leq 
\prod_{i=1}^k\rho(\mathcal{C}_i) \leq u^{\frac{3}{2} n(\ubft)-k}
\prod_{i=1}^kf(\mathcal{C}_i),
\end{equation}
where  $f(\mathcal{C})=\sum_{\gamma\in Span(\mathcal{C})}1$, that is, the number of spanning subgraphs contained 
in $\mathcal{C}$. 
Further,  we rewrite the partition function $\mathcal{Z}^{\beta,\ubft}_N$ as follows
$$\mathcal{Z}^{\beta,\ubft}_N= e^{\frac{3}{2}\beta n(\ubft)} (1-p)^{\frac{3}{2}n(\ubft)}\left(  2\rho(\Delta(\ubft)) +\sum_{k=2}^{n(\ubft)} 2^k 
\sum_{\{\mathcal{C}_1,\dots,\mathcal{C}_k\}\in \pi(\ubft)}\rho(\mathcal{C}_1)\cdots \rho(\mathcal{C}_k)\right),$$
where $\rho(\Delta(\ubft))=  \sum_{\gamma\in\; Span(\Delta(\ubft))} u^{|E(\gamma)|}$. \\

\noindent Now, observing that $e^{\frac{3}{2}\beta n(\ubft)} (1-p)^{\frac{3}{2}n(\ubft)}=e^{-\frac{3}{2}\beta n(\ubft)}$ and  employing 
the inequality (\ref{firstineq}), we obtain the following lower bound
\begin{equation}\label{1lowerb}
\begin{array}{ccl}
\mathcal{Z}^{\beta,\ubft}_N \!\!\!\!\!&\geq&\!\!\!\!\! e^{-\frac{3}{2}\beta n(\ubft)}\left\{  2\rho(\Delta(\ubft)) + 
u^{n(\ubft)} \displaystyle\sum_{k=2}^{n(\ubft)} \left(\frac{2}{u}\right)^{k}\!\!\!\!\!\!
\displaystyle\sum_{\{\mathcal{C}_1,\dots,\mathcal{C}_k\}\in \pi(\ubft)}\prod_{i=1}^kf(\mathcal{C}_i)\right\}   
\end{array}
\end{equation}
Notice that, given any decomposition $\{\mathcal{C}_1,\dots,\mathcal{C}_k\}\in \pi(\ubft)$,  
$f(\mathcal{C}_i)\geq 1$ for all $i\in\{1,\dots,k\}$, where  equality holds only  if $\mathcal{C}_i$ is a spanning tree. 
Further, if the degree of freedom, $k$ is large enough,  there exists  positive proportion of decompositions 
$\{\mathcal{C}_1,\dots,\mathcal{C}_k\}$ of $\ubft$ with $f(\mathcal{C}_i)= 1$ for all $k$. Thus, the lower bound can be obtain as
$\prod_{i=1}^kf(\mathcal{C}_i) \geq 1$. For example, for $k=n(\ubft), n(\ubft)-1$ and $n(\ubft)-2$, all
elements of any decomposition are spanning trees.
Substituting the lower  bound  in  (\ref{1lowerb}), we get the following inequalities
$$
\begin{array}{ccl}
\mathcal{Z}^{\beta,\ubft}_N & \geq & e^{-\frac{3}{2}\beta n(\ubft)}\left\{  2\rho(\Delta(\ubft)) + 
u^{n(\ubft)} \displaystyle\sum_{k=2}^{n(\ubft)} \left(\frac{2}{u}\right)^{k}
\displaystyle\sum_{\{\mathcal{C}_1,\dots,\mathcal{C}_k\}\in \pi(\ubft)}1 \right\} \\
      &  &   \\
      & \geq & e^{-\frac{3}{2}\beta n(\ubft)}\left\{  2\rho(\Delta(\ubft)) + 
u^{n(\ubft)} \displaystyle\sum_{k=1}^{n(\ubft)-1} \left(\frac{2}{u}\right)^{k+1}
\binom{n(\ubft)-1}{k} \right\}\\
      &  &   \\
      & = & e^{-\frac{3}{2}\beta n(\ubft)}\left\{  2\rho(\Delta(\ubft)) + 
2u^{n(\ubft)-1}\left( \left( 1+\displaystyle\frac{2}{u} \right)^{n(\ubft)-1} -1\right) 
\right\}\\
      &  &   \\
      & = & e^{-\frac{3}{2}\beta n(\ubft)}\left\{  2\rho(\Delta(\ubft)) - 
2 u^{n(\ubft)-1} + 2\left(2+u\right)^{n(\ubft)-1} \right\}
\end{array}
$$
Now, the function $\rho(\Delta(\ubft))(u)$ can be rewritten as
$$
\begin{array}{ccl}
\rho(\Delta(\ubft))&=&
u^{\frac{3}{2} n(\bft)} + 
\displaystyle\sum_{\substack{\gamma\in\; Span(\Delta(\ubft)):\\|E(\gamma)|<\frac{3}{2} n(\ubft)}} 
u^{|E(\gamma)|}.
\end{array}
$$
If $\mu>1$ and  number of strip $N$ is large enough,  we can conclude  that the behaviour of the function 
$\rho(\Delta(\ubft))(u)$ can be described by the term $u^{\frac{3}{2} n(\ubft)}$.  In fact, it is easy to see  that 
$\rho(\Delta(\ubft))>u^{\frac{3}{2} n(\ubft)}$ for any triangulation $\ubft$.
Utilizing this result, we derive the following lower bound for the partition function $\mathcal{Z}^{\beta,\bft}_N$ of the Ising model 
on any triangulation $\ubft$ and $N\in\N$,
\begin{equation}\label{eq3.24}
\begin{array}{ccl}
\mathcal{Z}^{\beta,\bft}_N & \geq & 2e^{\frac{3}{2}\beta n(\bft)}p^{\frac{3}{2}n(\bft)} - 
2e^{\frac{3}{2}\beta n(\bft)}p^{n(\bft)-1} (1-p)^{\frac{1}{2}n(\bft)+1}\\
 & & \\
 & & +\; 2e^{\frac{3}{2}\beta n(\bft)}(2-p)^{n(\bft)-1} (1-p)^{\frac{1}{2}n(\bft)+1}.
\end{array}
\end{equation}
Implementation of identity (\ref{FK_ISING}), and lower bound  (\ref{eq3.24})  provides the following lower bound for the partition
function of the annealed model,
\begin{equation}\label{eq3.25}
\begin{array}{ccl}
\Xi_N(\beta,\mu) &\geq& 2Z_N\left(\mu-\displaystyle\frac{3}{2}\beta -\displaystyle\frac{3}{2}\ln p\right) - 2 Z_N\left( \mu-\displaystyle\frac{3}{2}\beta -\ln p\sqrt{1-p} \right) +\\
    & &  2 Z_N\left( \mu-\displaystyle\frac{3}{2}\beta -\ln(2-p)\sqrt{1-p}\right).
\end{array}
\end{equation}
Thus, utilizing the  asymptotic property given in (\ref{yamb-e14}), we obtain that the 
partition function $\Xi_N(\beta,\mu)$ exists only  if 
$$\mu > \frac{3}{2}\beta +\ln2 + \frac{3}{2}\ln\left(1-e^{-2\beta}\right) +\ln\left(\cos\frac{\pi}{N+1}\right)\;\;\mbox{and}\;\; \beta> \frac{\ln2}{2}.$$
Now, let us discuss the case where $N\to\infty$. 
\begin{prop}\label{prop3.2}
If $(\beta,\mu)\in\mathbb{R}^2_+$ such that  
$$\mu < \displaystyle\frac{3}{2}\beta +\ln2 + \frac{3}{2}\ln\left(1-e^{-2\beta}\right)\;\;\mbox{and}\;\; \beta> \frac{\ln2}{2},$$ 
then there exists $N_0\in\N$ such that  the partition function $\Xi_N(\beta,\mu)=+\infty$ whenever 
$N > N_0$.  
Moreover, the Gibbs distribution  $\mathbb{P}^{\beta,\mu}_N$ with periodic boundary conditions cannot be defined 
by the standard formula with  $\Xi_N(\beta,\mu)$ being a normalising denominator. Consequently, there is 
no limiting probability measure $\mathbb{P}^{\beta,\mu}$ as $N\to\infty$, which implies that 
for any  finite-dimensional  cylinder  $\mathcal{C}_{i_1,\dots,i_k}$,
$\mathbb{P}^{\beta,\mu}_N (\mathcal{C}_{i_1,\dots,i_k}) = 0$ whenever  $N>N_0\geq\max\{i_1,\dots,i_k\}$.
\end{prop}

\vspace{0.5cm}

\noindent{\it The case $u< 1$ $\big($equivalent to $\beta<\frac{\ln2}{2}$ $\big)$}. 
Similarly to the first case, utilizing  (\ref{inequalityclus})  we obtain the following inequalities
$$u^{\frac{3}{2}|\mathcal{C}| -1}  \leq u^{|E(\gamma)|} \leq u^{|\mathcal{C}|-1},$$
which hold for all spanning subgraphs $\gamma\neq \ubft$. Thus, if $\{\mathcal{C}_1,\dots,\mathcal{C}_k\}\in\pi(\ubft)$ and $k\geq 2$ we obtain 
\begin{equation}\label{secondineq}
u^{\frac{3}{2}n(\ubft)-k}\prod_{i=1}^kf(\mathcal{C}_i)\leq 
\prod_{i=1}^k\rho(\mathcal{C}_i) \leq u^{n(\ubft)-k}
\prod_{i=1}^kf(\mathcal{C}_i).
\end{equation}
Employing the inequality (\ref{secondineq}), we obtain the 
lower bound for the partition function $\mathcal{Z}^{\beta,\ubft}_N $ on dual triangulation $\ubft$,
\begin{equation}\label{2lowerb}
\begin{array}{ccl}
\mathcal{Z}^{\beta,\ubft}_N \!\!\!\! & \geq &\!\!\!\!\!  e^{-\frac{3}{2}\beta n(\ubft)}\! \left\{ \! 2\rho(\Delta(\ubft))\! + \!
u^{\frac{3}{2} n(\ubft)} \displaystyle\sum_{k=2}^{n(\ubft)} \left(\frac{2}{u}\right)^{k}
\!\!\!\!\!\! \displaystyle\sum_{\{\mathcal{C}_1,\dots,\mathcal{C}_k\}\in \pi(\ubft)}\prod_{i=1}^kf(\mathcal{C}_i)\!\right\}.   
\end{array}
\end{equation}
Implementation of the same tecnique  as in the previous case yields the lower bound
for the critical curve,
\begin{equation}\label{notimprov}
\mu < \displaystyle\frac{3}{2}\beta +\ln2 + \frac{1}{2}\ln\left(1-e^{-2\beta}\right)+\ln\left(1+e^{-2\beta}\right)
\;\;\mbox{and}\;\; \beta< \frac{\ln2}{2}.
\end{equation}

\vspace{0.4cm}
\begin{proof}[PROOF OF THEOREM \ref{theomain1}.]
The proof of Theorem \ref{theomain1} follows immediately from  Proposition \ref{prop3.2} and the lower bound given by the Griffiths inequality.
\end{proof}

\begin{rmk}
The lower bound given by (\ref{notimprov}) does not improve the lower bound given by the first Griffiths inequality.
\end{rmk}

\subsection{Proof of Theorem \ref{theomain2}}\label{Sect3.3}

For each $N\in \N$, we define the following sets in $\mathbb{R}^2_{+}$
\begin{eqnarray}
\Gamma_N &=&\{(\beta,\mu)\in\mathbb{R}^2_{+}:\;\bfK^N\quad \mbox{is a trace class in}\;\ell^2_{\rm{T-C}}\}, 
\end{eqnarray}
\begin{eqnarray}
\Gamma^{-} =\bigcap_{N\in\mathbb{N}}\Gamma_N \quad & \mbox{and} & \quad  \Gamma^{+} =\bigcup_{N\in\mathbb{N}}\Gamma_N. 
\end{eqnarray}
Note that $\Gamma^{-}\subset \Gamma_N \subset \Gamma^{+}$ for all $N\geq 1$, $\Gamma_N\uparrow \Gamma^{+}$, and that the 
finite-volume Gibbs measure $\mathbb P^{\beta,\mu}_N$  on the set $\Gamma_N $  exists.  We also define  the 
$N$-strip functions $\rho_N$  associated with the  
set $\Gamma_N$,  for each $N\geq 1$, as follows
\begin{equation}\label{f_n}
\rho_N(\beta) = \inf\{ \mu\in\mathbb{R}^2_+ : (\beta,\mu)\in  \Gamma_N  \}\quad\mbox{for}\quad \beta\geq 0.
\end{equation}
By the properties of the trace class operators (see \cite{Ringrose} for an overview), we deduce that $\{\rho_N\}$ is a monotone decreasing sequence of measurable functions 
such that $f_1(\beta)\leq \rho_N(\beta)$, for all $\beta>0$, where $f_1$ is defined in (\ref{f_1}). Thus, we prove the existence of 
the pointwise limit 
\begin{equation}\label{f_TC}
\rho_{\rm{T-C}}(\beta):=\lim_{N\to\infty} \rho_N(\beta)\quad\mbox{for}\quad\beta\geq 0.
\end{equation}
Another important fact is that the graph of the function $\rho_{T-C}$ provides  an  upper bound for the critical curve. This property is 
a consequence of the  spectral properties of the  operator $\bfK$ introduced in (\ref{trmatrix})(see \cite{HeAnYuZo:2013} for the details 
and \cite{Ringrose} for an overview) and the definition of the function $\rho_{T-C}$. 
Consequently, implementation the results of \cite{HeAnYuZo:2013} and the  condition of subcriticality (\ref{AM}), we  can prove  the 
following proposition.
\begin{prop}\label{prop3.1}
There exist $N_0\in\N$ such that  for  $N> N_0$, the following property of functions $f_N$ is  fulfilled: 
\begin{enumerate}
 \item If $0< \beta < \beta^*_2$, then 
 \begin{equation}\label{eq5.2}
  f_{N}(\beta) \leq \psi(\beta),
 \end{equation}
 where $\beta^*_2$ is a positive solution of $(\ref{beta2})$ and  $\psi$ is defined in $(\ref{functpsi})$.
 \item If  $\beta^*_2 \leq \beta <\infty$, then 
  \begin{equation}\label{eq5.3}
 f_{N}(\beta) \leq \displaystyle\frac{3}{2}\beta +2\ln 2.
 \end{equation}
\end{enumerate} 
\end{prop}
\noindent Combining (\ref{f_TC}), (\ref{eq5.2}), (\ref{eq5.3}), and letting $N\to\infty$, we obtain the desired upper bound for the 
limit function $f_{T-C}$
\begin{equation}\label{upperbound}
\left\{
\begin{array}{ccccl}
f_{\rm{T-C}}(\beta) &\leq& \psi(\beta) &\mbox{if}& 0<\beta<\beta^{*}_2\\
\\
f_{\rm{T-C}}(\beta) &\leq& \displaystyle\frac{3}{2}\beta+2\ln2 &\mbox{if}& \beta^{*}_2 \leq\beta<\infty.
\end{array}\right.
\end{equation}
\noindent Since that the graph of the function $f_{T-C}$ lies above of the critical curve, the right-hand side of (\ref{upperbound}) provides 
an upper bound for the critical curve.

\vspace{0.4cm}
\begin{proof}[PROOF OF THEOREM \ref{theomain2}.] The upper bound for the critical curve $\gamma_{cr}$ is a consequence of
inequality  (\ref{upperbound}). The lower bound is a consequence of Theorem \ref{theomain1}. This concludes the proof of the theorem.
\end{proof}


\section{Discussion}
In this article we present a step towards  improving the subcriticality domain for an Ising model
coupled to two-dimensional CDT  introduced in \cite{HeAnYuZo:2013}. In doing so we employ FK representation of the Ising model on 
causal triangulations and combinatorial approximation. In addition, we make a first step  towards determining the critical 
curve of the model. Numerical evidence shows that the model has phase transition  (see \cite{Ambjorn:1999gi}, \cite{Ambjorn:2008jg}, 
\cite{Benedetti:2006rv}), but  this fact has never been proved explicitly. In this article we presente
mathematical
evidence of existence of the critical curve by studying the infinite-volume free energy, and computing a region where the critical curve 
for the annealed 
model is possibly located (see Figure \ref{fig2}). The discussion  in Section \ref{Sect2.3} along with
Corollary \ref{mainresult3} suggest that free energy can be expresed as $\ln \Lambda(\mu -\phi(\beta))$, where the function $\phi$ 
satisfies $\lim_{\beta\to 0^+}\phi(\beta)=\ln2$. This result leads to  a plausible conjecture that the boundary 
of the subcritical 
domain coincides with 
the locus of points $(\beta,\mu)$ where $\mu=\phi(\beta)+\ln2$, however in order to obtain more precise conclusions a 
further investigation is required.

In addition, Theorem \ref{theomain1} and Theorem \ref{theomain2} show that with respect to the weak limit Gibbs 
measure $\mathbb{P}^{\beta,\mu}$, the average  diameter of any  fixed height $n$
is  bounded  in  the region $\mu>f_2(\beta)$ (by a constant independent of $n$) with probability 1, which is essentially a one 
dimensional random geometry. Thus,
we can expect that this subcritical behaviour can be extended to the region 
$\mu>\gamma_{cr}(\beta)$, and that any causal triangulation  on the critical curve   has Hausdorff dimension $d_H=2$, 
$\mathbb{P}^{\beta,\mu_{cr}}$-a.s.
Finally, if the annealed model has  this property, then the region where it presents a  phase transition is possibly located on the critical curve. 
This direction also requires  further research.

\subsection*{Acknowledgements.} 
I would like to thank Prof. Y. Suhov and Prof. A. Yambartsev  for very valuable discussions  and his encouragement. 
This work was supported by FAPESP (projects  2012/04372-7,  2013/06179-2 and 2014/18810-1). Further, the author thanks the IME at the 
University of S\~ao Paulo for warm  hospitality.

\providecommand{\href}[2]{#2}\begingroup\raggedright\endgroup

\end{document}